\documentclass[12pt]{amsart}

\usepackage{amsfonts,amsmath,amssymb,amsthm}
\usepackage{setspace,graphicx}
\usepackage{sgame}
\usepackage{bm}
\usepackage{ifthen}
\usepackage{verbatim}
\usepackage{color}

\theoremstyle{plain}
\theoremstyle{definition}
  \newtheorem{theorem}{Theorem}

  \newtheorem{claim}{Claim}
  
  \newtheorem{corollary}{Corollary}

  \newtheorem{proposition}{Proposition}
  
  \theoremstyle{remark}

\newcommand{\RR}{\mathbb{R}}

\newcommand{\textd}{\text{d}}
\newcommand{\bbr}{\mathbb{R}}

\newcommand{\bfone}{\mathbf{1}}

\newcommand{\supp}{\text{supp}}

\begin{document}

\thispagestyle{empty}
\title[]{Rental harmony with roommates$^{\dagger}$}
\author[]{Yaron Azrieli$^{*}$ and Eran Shmaya$^{**}$}
\thanks{$^{\dagger}$We thank Herv\'{e} Moulin for helpful discussions, and Chris Chambers and Rodrigo Velez for their comments on a previous draft of this paper.\\
$^{*}$Department of Economics, The Ohio State University. email: azrieli.2@osu.edu.\\
$^{**}$Kellogg School of Management, Northwestern University, and School of Mathematics, Tel Aviv University. email: e-shmaya@kellogg.northwestern.edu\\
Date: \today}
\maketitle

\singlespacing
\begin{abstract}
We prove existence of envy-free allocations in markets with heterogenous indivisible goods and money, when a given quantity is supplied from each of the goods and agents have unit demands. We depart from most of the previous literature by allowing agents' preferences over the goods to depend on the entire vector of prices. Our proof uses Shapley's K-K-M-S theorem and Hall's marriage lemma. We then show how our theorem may be applied in two related problems: Existence of envy-free allocations in a version of the cake-cutting problem, and existence of equilibrium in an exchange economy with indivisible goods and money.

\bigskip
\bigskip


\bigskip
\bigskip

\end{abstract}


\onehalfspacing
\section{Introduction}
A central concept in the literature on economic fairness is \emph{envy-freeness}~\cite{fol-67,var-74,mou-03} -- an allocation is envy-free if no agent prefers the share allocated to another agent over his own share. In this note we study existence of envy-free allocations when the goods to be allocated are indivisible and heterogenous, and when in addition there is one perfectly divisible good (e.g., money). We assume that each agent has a demand for only one of the indivisible goods and that there is a given quantity supplied of each good.

While there are many real-life examples that can fit into this framework, we will use for concreteness the terminology of room-assignment and rent-division: Several rooms with different characteristics and given capacities are available in a house, and the total rent for the house needs to be divided between the rooms. In this context, envy-freeness boils down to a market clearing condition: A price is assigned to each room such that when each agent chooses his favorite room (given the prices) supply exactly equals demand and the market clears. Following Su~\cite{su-99}, we call such a situation \emph{rental harmony}. Note that even though we use the terms `rooms' and `capacities', we do not make the assumption that the agents to whom a given good is allocated, whom we call \emph{roommates}, receive a joint ownership of the same physical object. Rather, a room with capacity $7$ stands for an indivisible good of which $7$ units are supplied, and the roommates represent the $7$ agents who received these units. When we say that the price of the room is $p$ we mean that each unit costs $p/7$.

There is quite a vast literature dealing with different aspects of this model. Some of the earlier works include \cite{alk-91}, \cite{mas-87}, \cite{sve-83} and \cite{tad-91}. Where we depart from most of the previous works (with a couple of exceptions -- see below) is in the type of preferences that agents may have. Namely, it has been assumed in earlier works that each agent's preferences are defined over room-price pairs, i.e., if $r,r'$ are two rooms with prices $p$ and $p'$ respectively, then each agent can say whether he prefers to get room $r$ at price $p$ or room $r'$ at price $p'$. In our model an agent's favorite room may be a function of the \emph{entire vector of prices}. Thus, asking whether an agent prefers $(r,p)$ to $(r',p')$ is not meaningful in our context, since the answer may depend on the prices of other rooms.

There are several reasons why this is an important generalization. First, there may be `rational' reasons for agents' preferences over rooms to be affected by the entire vector of prices. This may be the case, for instance, if we view the choice of a room as only part of a larger `consumption plan'. For a concrete example, assume that a forward looking agent needs to choose between three types of cars, say High ($H$), Intermediate ($I$) and Low ($L$), with corresponding prices $p_H>p_I>p_L$. If $p_H$ is very high then an agent's preferred option may be to buy type $I$ and hold it for a long period of time. But if $p_H$ is reduced then the agent may prefer to buy $L$ initially (saving a larger part of his budget) and upgrade to $H$ later on when he has accumulated more wealth. Thus, his choice shifted from $I$ to $L$ even though the prices of these cars did not change.

Prices can also affect preferences if there is incomplete information about the quality of the rooms, in which case prices may serve as a signaling device. For instance, real-estate prices in two neighboring suburbs may provide information about their relative qualities. An increase (or decrease) in the price of houses in one of them may therefore affect the desirability of the other. Another reason for a similar effect is when agents take into account the fact that prices affect choices of other agents. In such an interactive situation there are plausible scenarios in which the entire vector of prices influences agents' optimal choices, for example if the price of a neighboring room indicates the identity of its future inhabitants.

Another reason to consider such general preferences is that framing effects and other well-documented `behavioral biases' may be affecting choices in ways that the standard model cannot capture. For example, assume that rooms $A$ and $B$ have similar characteristics while room $C$ is very different from the other two. Assume further that at a given price vector $p$ with $p_A=p_B$ the agent's preferred choice is room $C$. If the price of $A$ increases then room $B$ may become more attractive as it offers similar value as room $A$ for a `bargain' price. The agent may then choose $B$ instead of $C$, even though the prices of these rooms have not changed.

The two papers that are closest to ours are \cite{str-80} and \cite{su-99}. They allow for preferences as general as in our model, but in both these papers the supply of each good (room) is one, i.e., the number of agents is equal to the number of goods. The proofs in both these papers rely on this latter assumption. In another recently related paper Velez~\cite{vel-13} studies envy-free allocations in a general model with externalities. The existence result in that paper is based on the argument of \cite{su-99}. This paper also makes the assumption that the supply of each good is one.

Our contribution relative to these works is threefold. First, we allow for `roommates', i.e., the supply of each of the indivisible goods in the market may be greater than one, so that agents may be allocated different units of the same good. This extends the applicability of the result to many markets of interest. Notice that there is no straightforward way to reduce the problem into one in which the number of agents and rooms is the same. The reason is that, given the generality of preferences we allow over goods, there is no way to lift a preference over goods to a preference over units. 

Second, our proof introduces a new tool to this literature. The proof relies on a topological result of Shapley \cite{sha-73} known as the K-K-M-S Theorem. Roughly speaking, our proof works as follows. For each subset of rooms $T$ we consider the set of price vectors at which the demand for rooms in $T$ is sufficient to meet the capacity of these rooms. Our assumptions imply that for each $T$ this is a closed set, and that every collection of these sets corresponding to a balanced collection of subsets of rooms (see Appendix \ref{sec-kkms} for the definition) covers the simplex of all possible price vectors. It then follows from the K-K-M-S theorem that there exists a price vector in which the demand for every subset of rooms is sufficient. By Hall's marriage lemma it is then possible to assign rooms to agents to exactly clear the market. Previous papers~\cite{abd-04,dem-86} have used Hall's marriage lemma for this purpose, but to our knowledge the use of the K-K-M-S Theorem to establish the conditions required to apply the lemma is new, and as we show allows us to get a substantial generalization.

Third, we show equivalence between the rental harmony environment and two other problems: The cake division or chore division problem and a model of a discrete exchange economy with money introduced by Gale \cite{gal-84}. While these problem were studied in the past using similar mathematical tools, there seems to be no direct argument for equivalence between them in the previous literature. Thus, we achieve a generalization of the known results to these three problems and also establish the connection between them.

In the next section we state and prove our main result. In Section \ref{sec-variations} we show how our theorem may be applied in the two related environments of cake/chore division and a discrete exchange economy. Section \ref{sec-comments} concludes with some final remarks.

\section{Theorem and proof}
Let $N=\{1,2,\ldots,n\}$ be the set of \emph{agents} and let $R$ be the finite set of available \emph{rooms}. For each $r\in R$ let $c[r]>0$ be a positive integer representing the \emph{capacity} of room $r$. We assume that $\sum_{r\in R} c[r] =n$ (the case $\sum_{r\in R} c[r] >n$ trivially follows). Let
$$F=\left\{f:N\to R ~:~ \left|f^{-1}(r)\right|=c[r] ~ \forall r\in R \right\}$$
be the set of all assignments of agents to rooms that respect the capacity constraints.

The total rent for the house is normalized to 1, and we let
$$\Delta(R) = \left\{ \{p[r]\}_{r\in R} ~:~ \sum_{r\in R} p[r] =1,~ p[r]\ge 0 ~ \forall r\in R \right\}$$
be the set of possible ways to allocate the rent among the different rooms. We view $\Delta(R)$ as a subset of $\RR^R$ and work with the standard topology it inherits from that space. For $p\in \Delta(R)$ the support of $p$ is the set $\supp(p)=\{r\in R ~:~ p[r]>0\}$. If $T\subseteq R$ then $\Delta(T)= \{p\in \Delta(R) ~:~ \supp(p)\subseteq T \}$ is the face of $\Delta(R)$ corresponding to $T$.

Given a price vector $p$, each agent $i$ has a set $L_i(p)\subseteq R$ of rooms she likes most at these prices. We assume\\
(A1) For each $i$ and $p$, $L_i(p)\neq \emptyset$.\\
(A2) For each $i$ and $p$, $\supp(p)^c \subseteq L_i(p)$.\\
(A3) For each $i$, $L_i$ has a closed graph. That is, $\{p\in \Delta(R) ~:~ r\in L_i(p)\}$ is closed for every $r\in R$ and every $i\in N$.

Assumption (A1) requires that every agent likes at least one of the rooms given each price vector. (A2) says that all agents like free rooms. We elaborate on this assumption in Section
\ref{subsec-A2}. Finally, (A3) 
reflects continuity of preferences in prices.

\begin{theorem}\label{thm-rental-harmony}
Under assumptions (A1), (A2) and (A3) there exists $p^*\in \Delta(R)$ and an assignment $f^*\in F$ such that $f^*(i)\in L_i(p^*)$ for every $i\in N$.
\end{theorem}

\begin{proof}
For every $p\in \Delta(R)$ and every $T\subseteq R$ let
$$A_T(p) = \left\{i\in N ~:~ L_i(p)\cap T \neq \emptyset\right\}$$
be the set of agents who like one of the rooms in $T$ at prices $p$. Also, for $T\subseteq R$ define
$$K_T = \left\{p\in \Delta(R) ~:~ |A_T(p)| \ge \sum_{r\in T} c[r]\right\}$$
to be the set of price vectors at which the demand for rooms in $T$ is sufficient to meet the capacity of these rooms.

\begin{claim}\label{claim-closed}
Each $K_T$ is closed in $\Delta(R)$.
\end{claim}
\begin{proof}
Note that
\[K_T=\bigcup_{B,g}\bigcap_{i\in B} \{p\in\Delta(R):g(i)\in L_i(p)\},\]
where the union ranges over all pairs $(B,g)$ such that $B$ is a set of agents with $|B|\ge \sum_{r\in T} c[r]$ and $g:B\rightarrow T$ is an assignment of rooms in $T$ to the agents in $B$.

The sets $\{p\in\Delta(R):g(i)\in L_i(p)\}$ are closed for every $i$ and $g$ by (A3). Therefore, $K_T$ is closed as a finite union of intersections of closed sets.
\end{proof}

\begin{claim}\label{claim-balanced}
If $\mathcal{T}$ is a balanced collection of subsets of $R$ (see Appendix \ref{sec-kkms} for the definition) then $\bigcup_{T\in \mathcal{T}} K_T = \Delta(R)$
\end{claim}
\begin{proof}
Let $\mathcal{T}$ be a balanced collection and let $\{\lambda_T\}_{T\in \mathcal{T}}$ be non-negative coefficients satisfying $\sum_{T\in\mathcal{T}} \lambda_T \bfone_T = \bfone_R$. Taking scalar product with arbitrary $u\in\bbr^R$ we have that
\begin{equation}\label{eqn-balanced}
\sum_{T\in\mathcal{T}} \lambda_T \sum_{r\in T} u[r] = \sum_{r\in R} u[r].
\end{equation}

Fix some $p\in \Delta(R)$ and let $g:N\to R$ be a choice of optimal rooms for the players at prices $p$ (here we use (A1)), so that in particular $g^{-1}(T) \subseteq A_T(p)$ for every $T\in \mathcal{T}$. Then
\begin{eqnarray*}
\sum_{T\in\mathcal{T}} \lambda_T |A_T(p)| \ge \sum_{T\in\mathcal{T}} \lambda_T |g^{-1}(T)| &=& \sum_{T\in\mathcal{T}} \lambda_T \sum_{r\in T} |g^{-1}(r)| =\\
& & \sum_{r\in R} |g^{-1}(r)| = n = \sum_{r\in R} c[r] = \sum_{T\in\mathcal{T}} \lambda_T \sum_{r\in T} c[r],
\end{eqnarray*}
where the second and last equalities follow from (\ref{eqn-balanced}) with $u[r]=|g^{-1}(r)|$ and $u[r]=c[r]$, respectively. It follows that there is $T\in \mathcal{T}$ such that $|A_T(p)| \ge \sum_{r\in T} c[r]$, so that $p\in K_T$.
\end{proof}

It follows from Claims \ref{claim-closed} and \ref{claim-balanced} that the collection of sets $\{K_T\}_{T\subseteq R}$ satisfies the conditions of Corollary \ref{cor-KKMS} in Appendix \ref{sec-kkms}. Thus, there exists $p^*\in \Delta(R)$ such that $p^*\in \bigcap_{T\subseteq \supp(p^*)} K_T$.

Now, consider a bipartite graph with sets of vertices $N$ and $R$, where a node $i\in N$ is connected to a node $r\in R$ if $r\in L_i(p^*)$. If $T\subseteq \supp(p^*)$ then $|A_T(p^*)| \ge \sum_{r\in T} c[r]$ since $p^*\in K_T$, and if $T\nsubseteq \supp(p^*)$ then $|A_T(p^*)|=n \ge \sum_{r\in T} c[r]$ since all the players like free rooms. It follows that the graph satisfies the condition of Hall's Marriage Theorem (See Theorem \ref{thm-marriage}), so there is a subgraph in which each agent is connected to at most one of the rooms in $R$ and each room in $R$ is exactly full. Since $\sum_{r\in R} c[r] =n$ each agent is connected to exactly one room. This defines the required assignment $f^*$.

\end{proof}


\section{Variations of the problem}\label{sec-variations}

 \subsection{Cake cutting and chore division}
A closely related problem to the one we consider is the problem of allocating pieces of a cake to a group of agents in a way that every agent is happy with the piece he got. There are several formulations of this problem, starting with the classic works \cite{dub-61} and \cite{ste-49}. In the version closest to our model (see, e.g., \cite[Section 3]{su-99}) the cake has a rectangular shape and one can only use $n-1$ vertical cuts to partition the cake into $n$ pieces ($n$ is the number of agents). Each possible cake-cut corresponds then to a point in the $n-1$ dimensional simplex. If players have their favorite piece(s) given  any cake-cut, and one wants to find a cut in which each player likes a different piece, then this fits to the special case of our model in which $n=|R|$ and $c[r]=1$ for each $r\in R$.

A similar problem is that of chore division, in which a set of undesirable entities (`chores') is to be allocated to a group of agents. Each chore comes with a  monetary compensation attached to it as well as the number of agents that should be performing it. One is interested in finding compensations for the various chores such that when each agent chooses a favorite chore there are enough agents performing each chore. One example of this situation would be the allocation of administrative tasks to faculty members in an academic department.

What is common to both these problems, and different from the rental harmony problem we considered, is that higher amounts of the divisible good are desired by the agents. In the rental harmony problem we interpreted the transfers as rent that an agent pays for his room, and we assumed in (A2) that agents like free rooms. On the other hand, in the cake cutting problem $p[r]=0$ means that the $r$th piece is empty, and so a hungry agent would not want to get it. Similarly, in the chore division problem a chore without compensation is unlikely to be the favorite of any agent.

Consider the following alternative to (A2), which requires that agents \emph{never} like an empty piece of cake (or a chore with no compensation):

\medskip
\noindent (A2$^*$)  For each $i$ and $p$, $L_i(p) \subseteq \supp(p)$.
\medskip

\begin{proposition}\label{pro:cake}
Under assumptions (A1), (A2$^*$) and (A3) there exists $p^*\in \Delta(R)$ and an assignment $f^*\in F$ such that $f^*(i)\in L_i(p^*)$ for every $i\in N$.
\end{proposition}
\begin{proof}
Consider preferences $L_i^\ast$ over $\Delta(R)$ that satisfy (A1), (A2$^*$) and (A3). We transform these preferences to preferences $L_i$ over $\Delta(R)$ that satisfy (A1), (A2) and (A3) in the following way: For each room $r\in R$ let $v_r$ be the vertex of $\Delta(R)$ corresponding to that room, and let $F_r=\Delta(R\setminus\{r\})$ be the face of $\Delta(R)$ opposite to $v_r$. 
Denote by $w_r$ the barycenter of $F_r$, that is $w_r=\frac{1}{|R|-1}\sum_{s\in R\setminus\{r\}} v_s$.
Let $\varphi:\Delta(R)\rightarrow\Delta(R)$ be the unique affine embedding such that $\varphi(v_r)=w_r$. Then $\varphi$ maps $\Delta(R)$ onto a smaller copy of this simplex, which lies inside $\Delta(R)$. In particular, $\varphi$ maps the boundary of $\Delta(R)$ onto the boundary of $\varphi(\Delta(R))$, and the interior of $\Delta(R)$ onto the interior of $\varphi(\Delta(R))$. 

Define $L_i$ by
\[L_i(p)=\begin{cases}L^\ast_i (\varphi^{-1}(p)), &\text{ if }p\in \text{interior}(\varphi(\Delta(R))),\\
L^\ast_i(\varphi^{-1}(p))\cup \{r\in R ~:~ p[r]\le 1/|R|\}, &\text{ if }p\in \text{boundary}(\varphi(\Delta(R))),\\
\{r\in R ~:~p[r]\le 1/|R|\}, &\text{ otherwise}.\end{cases}\]
In words: If $p$ is in the interior of the image of $\varphi$ (the interior of the small copy of the simplex) then the favorite rooms under $L_i$ are the same as those under $L_i^\ast$ when prices are $\varphi^{-1}(p)$; if $p$ is not in the image of $\varphi$ then only the relatively cheap rooms are the favorites; if $p$ is on the boundary of the image of $\varphi$ then both rooms that are favorite under $L^\ast_i$ when prices are $\varphi^{-1}(p)$ and the relatively cheap rooms are preferred.


It is straightforward to verify that $L_i$ satisfies assumptions (A1), (A2) and (A3). It follows from Theorem \ref{thm-rental-harmony} that there exists an envy-free allocation for these preferences. Let $\bar p$ be the price vector associated with this allocation. We claim that $\bar p \in \text{interior}(\varphi(\Delta(R)))$. To see why notice first that $\bar p$ must be in the image of $\varphi$, since outside the image only relatively cheap rooms are liked, and there is always at least one room with price greater than $1/|R|$ that no agent would choose. Second, assume by contradiction that $\bar p$ is on the boundary of the image of $\varphi$. Then $\bar p\in \varphi (F_r)$ for some room $r$. But then no agent likes room $r$ at prices $\bar p$ since by (A2$^\ast$) $r\notin L^\ast_i(\varphi^{-1}(\bar p))$ and since $\bar p[r]=1/(|R|-1)>1/|R|$.

To conclude, $\bar p \in \text{interior}(\varphi(\Delta(R)))$ and therefore $L_i(\bar p) = L^\ast_i (\varphi^{-1}(\bar p))$ for each agent $i$. It follows that there is an envy-free allocation for preferences $L^\ast_i$ with prices $\varphi^{-1}(\bar p)$.
\end{proof}

\subsection{Equilibrium in a discrete exchange economy}
Our result can be used to prove existence of equilibrium in an exchange economy with indivisible goods and money, as in the model studied by Gale \cite{gal-84}.\footnote{We thank Rodrigo Velez for pointing us to this Gale paper.} While Gale assumes that supply of each of the indivisible goods is 1, we allow for arbitrary quantities. The essential difference between the rental harmony problem we consider and Gale's exchange economy is that the prices of the rooms need not sum up to 1. Instead, it is only assumed that the price of each room is non-negative and bounded above (by 1, without loss). Thus,
$$C(R) = [0,1]^R$$
is the set of possible ways to price the different rooms. To stay consistent with our previous terminology we keep calling the indivisible goods `rooms', even though the interpretation of the model is somewhat different now. Agents' preferences are still represented by the sets $L_i(p)\subseteq R$, so that $r\in L_i(p)$ means that agent $i$ likes room $r$ at prices $p$. We keep assumptions (A1) and (A3) unchanged, but we replace the monotonicity assumption (A2) with the following arguably more compelling assumption:

\medskip
\noindent (A2') For each $i$ and $p$, $L_i(p)\subseteq \{r\in R ~:~ p[r]<1\}$.
\medskip

Thus, instead of assuming that one of the rooms with price $0$ will be chosen we assume that a room with price $1$ (the maximal possible price) will not be chosen. 
We also weaken (A1) to allow for the possibility that no room is desirable when the prices of all rooms is $1$. As should be clear from the proof we could weaken (A1) further.

\medskip

\noindent (A1') For each $i$ and $p$ such that $p[r]\neq 1$ for some $r\in R$, $L_i(p)\neq\emptyset$.
\medskip

\begin{proposition}
Under assumptions (A1'), (A2') and (A3) there exists $p^*\in C(R)$ and an assignment $f^*\in F$ such that $f^*(i)\in L_i(p^*)$ for every $i\in N$.
\end{proposition}
\begin{proof}
Let $B(R)=\{p\in C(R) ~:~ p[r]=0\text{ for some } r\}$ and let $\varphi:B(R)\leftrightarrow \Delta(R)$ be a homeomorphism with the property that $\varphi(p)[r]=0$ whenever $p[r]=1$. Such a homeomorphism is constructed in Gale's proof of his theorem. Then $\varphi$ transforms preferences over $B(R)$ that satisfy (A1'), (A2') and (A3) into preferences over $\Delta(R)$ that satisfy (A1), (A2$^\ast$) and (A3). By Proposition~\ref{pro:cake} these preferences admits an envy-free allocation.\end{proof}

%

\section{Final comments}\label{sec-comments}

\subsection{On assumption (A2)}\label{subsec-A2}
Assumption (A2) is probably the most restrictive of our conditions. 
Because we did not assume that agents' preference are monotonic in the prices, (A2) is the only assumption that captures the intuition that agents are tightfisted: Free rooms are always at least as good as non-free rooms

It is possible to relax (A2) somewhat without affecting the result. Consider the following assumption:

\medskip
(A2$^\circ$) If $\supp(p)\neq R$ then
$\supp(p)^c \cap L_i(p) \neq \emptyset$ for every agent $i$.
\medskip

This weaker version requires that if there are free rooms then every agent likes at least one of them. Our result holds unchanged if (A2) is replaced by (A2$^\circ$). The reason is that, given (A3), (A2$^\circ$) implies (A2). To see why, fix some $p$ with $|\supp(p)^c|\ge 2$ (if $|\supp(p)^c|\le 1$ then there is nothing to prove), and let $\bar r\in \supp(p)^c$. For every $\alpha\in (0,1)$ consider the price vector $p_\alpha$ defined by $p_\alpha[\bar r]= p[\bar r]=0$ and $p_\alpha[r] = \alpha p[r] +(1-\alpha)\frac{1}{|R|-1}$ for each $r\neq \bar r$. Then $\bar r$ is the only free room at every $p_\alpha$, so by (A2$^\circ$) every agent likes $\bar r$. But $p_\alpha \to p$ as $\alpha \to 1$, so by (A3) every agent likes $\bar r$ at prices $p$ as well.

Still, even this weaker version rules out many standard preferences. In particular, any quasi-linear preferences in which two rooms have different values do not satisfy this assumption, and indeed envy-free allocations need not exist with such preferences (see, for example, \cite[Section 6]{abd-04}). However, starting from any preferences $L_i$ satisfying (A1) and (A3) (in particular, quasi-linear preferences), it is possible to obtain modified preferences that satisfy all three assumptions by altering $L_i$ only on the boundary of the simplex. Specifically, the correspondence $\tilde{L}_i(p)= L_i(p) \cup \supp(p)^c$ satisfies (A1), (A2) and (A3) whenever $L_i$ satisfies (A1) and (A3).

\subsection{Efficiency}
Some of the previous papers on fair allocations have studied the relationship between envy-freeness and efficiency. In our model, however, it is not clear what efficiency means. The reason is that preferences of agents are defined over the indivisible goods (rooms) \emph{conditional} on the vector of prices. Thus, one cannot compare allocations across different price vectors.

An alternative approach, which allows to consider efficiency, would be to start from preferences over pairs $(p,r)$ where $p$ is the vector of prices and $r$ is the room assigned to the agent.\footnote{This is essentially the approach taken in \cite{vel-13}. In \cite{vel-13} preferences are defined over the entire allocation (including how rooms are assigned to other agents), but the {impersonality} axiom implies that agents only care about the vector of prices and the room assigned to them.} From such preferences one can derive the preferences over rooms conditional on prices. However, such across-prices comparisons are not relevant for the question of existence of envy-free allocations, which is the focus of this note. We therefore preferred to simplify the exposition and notation by using the conditional preferences as primitive.

\subsection{Manipulation}
While we proved existence of envy-free allocations, we did not study whether such allocations can be implemented when agents' preferences are their private information. This aspect of the problem has been analyzed under the type of preferences allowed in the previous literature -- see for example \cite{vel-11} and the references therein. It would be interesting to see which of the results obtained in that literature apply in our set-up as well.

\subsection{Constructing a solution}
Our proof is not constructive, as the K-K-M-S theorem guarantees existence of the desired price vector $p^*$ without showing how to find it. However, one could construct algorithms that approximate $p^*$ up to an arbitrary level of precision. For instance, Shapley's original proof of the K-K-M-S theorem relies on a combinatorial result in the spirit of Sperner's lemma, which can be used approximate $p^*$. Once $p^*$ is found it is easy to construct the envy-free assignment $f^*$.

\appendix

\section{K-K-M-S Theorem}\label{sec-kkms}
A collection $\mathcal{T}$ of subsets of $R$ is called \emph{balanced} if there are non-negative coefficients $\{\lambda_T\}_{T\in \mathcal{T}}$ such that
$$\sum_{T\in\mathcal{T}} \lambda_T \bfone_T = \bfone_R.$$

The following result by Zhou \cite{zho-94-2} is a variant of Shapley's \cite{sha-73} `K-K-M-S theorem', the only difference being that the covering sets are open rather than closed. See also \cite{bil-70,her-97,ich-88,kom-94,sha-91,zho-94-1,zho-94-2} for alternative proofs of the K-K-M-S theorem and related results, as well as for applications of this result in the theory of cooperative games.
\begin{theorem}\label{open-KKMS}\cite{zho-94-2}
Let $\{L_T\}_{T\subseteq R}$ be a collection of open subsets of $\Delta(R)$ with the property that $\Delta(S) \subseteq \bigcup_{T\subseteq S} L_T$ for every $S\subseteq R$. Then there exists a balanced collection $\mathcal{T}$ such that $\bigcap_{T\in \mathcal{T}} L_T \neq \emptyset$.
\end{theorem}

\begin{corollary}\label{cor-KKMS}
Let $\{K_T\}_{T\subseteq R}$ be a collection of closed subsets of $\Delta(R)$ such that $\bigcup_{T\in \mathcal{T}} K_T = \Delta(R)$ whenever $\mathcal{T}$ is a balanced collection. Then there is $p^*\in\Delta(R)$ such that $p^*\in \bigcap_{T\subseteq \supp(p^*)} K_T$.
\end{corollary}

\begin{proof}
For each $T$ define $L_T=K_T^c$. Then each $L_T$ is open and $\bigcap_{T\in \mathcal{T}} L_T = \emptyset$ for every balanced collection $\mathcal{T}$. By Theorem \ref{open-KKMS} there is $S\subseteq R$ and $p^*\in \Delta(S)$ such that $p^*\not\in L_T$ for every $T\subseteq S$. Thus, $p^*\in \bigcap_{T\subseteq S} K_T \subseteq \bigcap_{T\subseteq \supp(p^*)} K_T$.
\end{proof}



\section{Marriage Theorem with polygamy}\label{sec-marriage}
\begin{theorem}\label{thm-marriage}~\cite[Corollary 3.11]{bol-98}
Let $G$ be a bipartite graph with vertex sets $X$ and $Y$, and let $c:Y\to \mathbb{N}$. Then $G$ contains a subgraph $H$ such that $\textd_H(y)=c[y]$ for every $y\in Y$ and $\textd_H(x)\in \{0,1\}$ for every $x\in X$ if and only if for every $S\subseteq Y$
\[\textd_G(S)\geq \sum_{y\in S}c[y].\]
\end{theorem}



\begin{thebibliography}{99}


\bibitem{abd-04} Abdulkadiroglu, A., T. S\"{o}nmez and M. U. \"{U}nver, ``Room assignment-rent division: A market approach,'' \emph{Social Choice and Welfare} 22 (2004), 515-538.


\bibitem{alk-91} Alkan, A., G. Demange and D. Gale, ``Fair allocation of indivisible goods and criteria for justice,'' \emph{Econometrica} 59 (1991), 1023-1039.


\bibitem{bil-70} Billera, L. J., ``Some theorems on the core of an $n$-person game without side payments,'' \emph{Siam Journal of Applied Mathematics} 18 (1970), 567-579.


\bibitem{bol-98} Bollob{\'a}s, B., \emph{Modern graph theory} (New York: Springer-Verlag, 1998).


\bibitem{dem-86} Demange, G., D. Gale and M. Sotomayor, ``Multi-item auctions,'' \emph{The Journal of Political Economy} 94 (1986), 863-872.


\bibitem{dub-61} Dubins, L. E. and E. H. Spanier, ``How to cut a cake fairly,'' \emph{The American Mathematical Monthly} 68 (1961), 1-17.


\bibitem{fol-67} Foley, D., ``Resource allocation and the public sector,'' \emph{Yale Economic Essays} 7 (1967), 45-98.


\bibitem{gal-84} Gale, D., ``Equilibrium in a discrete exchange economy with money,'' \emph{International Journal of Game Theory} 13, 61-64.


\bibitem{her-97} Herings, P. J. J., ``An extremely simple proof of the K-K-M-S theorem,'' \emph{Economic Theory} 10 (1997), 361-367.


\bibitem{ich-88} Ichiishi, T., ``Alternative version of Shapley's theorem on closed coverings of a simplex,'' \emph{Proceedings of the American Mathematical Society} 104 (1988), 759-763.


\bibitem{kom-94} Komiya, H., ``A simple proof of the K-K-M-S theorem,'' \emph{Economic Theory} 4 (1994), 463-466.


\bibitem{mas-87} Maskin, E., ``On the fair allocation of indivisible goods,'' in Feiwel, G., Ed. \emph{Arrow and the foundations of the theory of economic policy} (London: Macmillan, 1987), 341-349.


\bibitem{mou-03} Moulin, H., \emph{Fair division and collective welfare} (Cambridge: MIT press, 2003).


\bibitem{sha-73} Shapley, L. S., ``On balanced games without side payments,'' in Hu, T. C. and S. M. Robinson, Eds. \emph{Mathematical Programming} (New York: Academic Press, 1973), 261-290.


\bibitem{sha-91} Shapley, L. S. and R. Vohra, ``On Kakutani's fixed point theorem, the K-K-M-S theorem and the core of a balanced game,'' \emph{Economic Theory} 1 (1991), 108-116.


\bibitem{ste-49} Steinhaus, H., ``Sur las division pragmatique,'' \emph{Econometrica} 17 (1949), 315-319.


\bibitem{str-80} Stromquist, W., ``How to cut a cake fairly,'' \emph{The American Mathematical Monthly} 87 (1980), 640-644.


\bibitem{su-99} Su, F. E., ``Rental harmony: Sperner's lemma in fair division,'' \emph{The American Mathematical Monthly} 106 (1999), 930-942.


\bibitem{sve-83} Svensson, L., ``Large indivisibilities: An analysis with respect to price equilibrium and fairness,'' \emph{Econometrica} 51 (1983), 939-954.


\bibitem{tad-91} Tadenuma, K. and W. Thomson, ``No-envy and consistency in economies with indivisible goods,'' \emph{Econometrica} 59 (1991), 1755-1767.


\bibitem{var-74} Varian, H., ``Equity, envy, and efficiency,'' \emph{Journal of economic Theory} 9 (1974), 63-91.


\bibitem{vel-11} Velez, R. A., ``Are incentives against economic justice?,'' \emph{Journal of Economic Theory} 146 (2011), 326-345.


\bibitem{vel-13} Velez, R. A., ``Fairness and externalities,'' mimeo, Texas A\&M University, 2013.


\bibitem{zho-94-1} Zhou, L., ``A new bargaining set of an N-person game and endogenous coalition formation,'' \emph{Games and Economic Behavior} 6 (1994), 512-526.


\bibitem{zho-94-2} Zhou, L., ``A theorem on open coverings of a simplex and Scarf's core existence theorem through Brouwer's fixed point theorem,'' \emph{Economic Theory} 4 (1994), 473-477.


\end{thebibliography}
\end{document}